\documentclass[11 pt]{article}
\usepackage{cite}
\usepackage{comment}
\usepackage{amsthm}

\newcommand{\KV}{\mathcal{K}\mathcal{V}}
\newtheorem*{thm}{Theorem}

\usepackage{amsfonts}
\usepackage{amsmath}
\usepackage[mathscr]{eucal}
\usepackage{amssymb}
\usepackage{verbatim}
\usepackage{enumerate}
\usepackage{epsfig}

\newtheorem{theorem}{Theorem}
\newtheorem{lemma}[theorem]{Lemma}

\newtheorem{corollary}[theorem]{Corollary}

\newtheorem{claim}[theorem]{Claim}

\newtheorem{conjecture}[theorem]{Conjecture}

\newtheorem{definition}[theorem]{Definition}
\newtheorem{observation}[theorem]{Observation}

\advance \topmargin by -\headheight
\advance \topmargin by -\headsep

\evensidemargin \oddsidemargin
\newcommand\bz{\mbox{$\mathbb Z$}}

\newcommand{\marginlabel}[1]%
{\mbox{}\marginpar{\it{\raggedleft\hspace{0pt}#1}}}

\newcommand{\norm}[1]{\left\lVert #1 \right\rVert}

\def\implies{\Rightarrow}

\newlength{\pgmtab}  
\newcounter{lecnum}

\newlength{\tpush}
\newcommand{\MC}{{\sc Maximum Cut}}

\newcommand{\SC}{{\sc Sparsest Cut}}

\newcommand{\onev}{\textbf{1}} 

\begin{document}

\title{Spectral Algorithms for Unique Games}
\author{Alexandra Kolla\thanks{Microsoft Research.}}
\date{}

\maketitle

\begin{abstract}
We give a new algorithm for Unique Games which is based on purely {\em spectral} techniques, in contrast to previous
work in the area, which relies heavily on semidefinite programming (SDP). Given a highly satisfiable instance of Unique Games, our algorithm is able to recover a good assignment.
The approximation guarantee depends only on the completeness of the game, and not on the alphabet size,
while the running time depends on spectral properties of the {\em Label-Extended} graph associated with the instance of Unique Games.

We further show that on input the integrality gap instance of Khot and Vishnoi, our algorithm runs in quasi-polynomial time and decides that the instance if highly unsatisfiable. Notably, when run on this instance, the standard SDP relaxation of Unique
Games {\em fails}. As a special case, we also re-derive a polynomial time algorithm for Unique Games on
expander constraint graphs.

The main ingredient of our algorithm is a technique to effectively use the full spectrum of the underlying graph instead of just the second eigenvalue, which is of independent interest. The question of how to take advantage of the full spectrum of a graph in the design of algorithms has been often studied, but no significant progress was made prior to this work.

\end{abstract}

\newpage

\section{Introduction}

A Unique Game is defined in terms of a constraint graph $G = (V,E)$, an integer $k$ which is called the alphabet,
a set of variables $\{x_u\}_{u \in V}$, one for each vertex $u$, and
a set of permutations (constraints) $\pi_{uv}: [k] \rightarrow [k]$,
one for each edge $(u,v)\in E$. An assignment to the variables (labeling) is said to
satisfy the constraint on the edge $(u,v) \in E$ if $\pi_{uv}(x_u) = x_v$.
The edges are taken to be undirected and hence $\pi_{uv} = (\pi_{vu})^{-1}$.
The goal is to assign a value from the
set $[k]$ to each variable $x_u$ so as to maximize the number of
satisfied constraints.

Khot \cite{KhotUCSP} conjectured that it is NP-hard
to distinguish between the cases when almost all the constraints of a Unique
Game are satisfiable and when very few of the constraints are satisfiable:

\begin{conjecture}(Unique Games Conjecture-{\sf UGC})
For any constants $\epsilon, \delta > 0$, there is a $k(\epsilon,\delta)$ such that for any $k > k(\epsilon,
\delta)$, it is NP-hard to distinguish between instances of Unique
Games with alphabet size $k$ where at least $1-\epsilon$ fraction of
constraints are satisfiable and those where at most $\delta$
fraction of constraints are satisfiable.
\end{conjecture}

There are two reasons which make {\sf UGC} particularly intriguing. First, it is a well-balanced question. Despite continuous efforts to prove or disprove it, there is still no consensus regarding its validity. This seems to indicate that {\sf UGC} is more likely to be resolved in the near future in contrast to the {\sf P} vs. {\sf NP} problem for example, for which it is widely believed that {\sf P}$\neq${\sf NP} but current techniques have been unable to prove it. Second, as seen by a series of works, the truth of {\sf UGC} implies that the currently best known approximation algorithms for
many important computational problems have optimal approximation ratios.
Since its origin, {\sf UGC} has been successfully used to prove
often optimal hardness of approximation results for several
important {\sf NP}-hard problems such as {\sc Min-2Sat-Deletion}
\cite{KhotUCSP}, {\sc Vertex Cover} \cite{KhotRegev}, {\MC}
\cite{KKMO} and non-uniform {\SC}
\cite{CKKRS,KV}. In addition, in recent years, {\sf UGC}
has also proved to be intimately connected to the limitations of
semidefinite programming (SDP). Making this connection precise, the authors in \cite{Aus} and \cite{R} show that if {\sf UGC} is true, then for every constraint satisfaction
problem (CSP) the best approximation ratio is given by a
certain simple SDP.

 Arguably, a seemingly strong reason
for belief in {\sf UGC} is the failure of several attempts to design efficient algorithms for Unique Games using current state-of-the-art techniques, such as linear and semidefinite programming (LP/SDP). Indeed, several works \cite{KV,RS} show limitations of LPs and SDPs in solving Unique Games by exhibiting the existence of \textit{itegrality gap} instances. Those are Unique Games instances that {\em fool} the SDP since,
even though they have no good satisfying assignment, the respective SDP solution is high. The existence of such instances implies, in particular,
that there is no hope for a ``good'' approximation algorithm based on those SDPs. Moreover, recently it was shown that solving Unique Games is at least as hard as another seemingly hard problem called the small set expansion problem \cite{RS10}. Our work presents evidence that a different set of techniques might be more powerful when it comes to designing algorithms for Unique Games thus giving hope that such techniques can be used to potentially disprove {\sf UGC}.

We present a purely {\em spectral} algorithm for Unique Games that finds highly
satisfying assignments when they exist. The running time of our algorithm depends on spectral properties of the {\em Label-Extended} graph or the constraint graph associated with the instance of Unique Games.
Our algorithm runs in polynomial or quasi-polynomial time for a large class of instances, including those where the standard SDP
provably fails as well as instances where the constraint graph is an expander. At a high level, we show that given $\epsilon >0$,
 there is a $\delta=\delta(\epsilon)$ such that the algorithm is able to distinguish between the following two cases:
 \begin{itemize}
 \item \textbf{YES} case: There exists an assignment that satisfies $(1-\epsilon)$ fraction of the constraints.
 \item \textbf{NO} case: Every assignment satisfies less that $\delta$ fraction of the constraints.
 \end{itemize}

  In particular, our algorithm runs in quasi-polynomial time on input the SDP integrality gap instance as it appears in \cite{KV} (hereafter referred to as $\KV$). We note that the authors in \cite{KV} roughly showed that, when run on a certain
highly unsatisfiable instance of Unique Games, the standard SDP relaxation has very high
objective value and consequently fails to distinguish between the two cases above. As another special case, our algorithm runs in polynomial time when the constraint graph is an expander (therefore re-deriving results similar to \cite{AKKTSV} and \cite{KT}). Moreover, similarly to
\cite{AKKTSV,MM,KT,AIMS}, the performance of the algorithm does not depend on the alphabet size $k$.

Our main result is the following:

\begin{thm}(Main)\label{thm:main0}
Let $\mathcal{U}=(G,M,k)$ be a $(1-\epsilon)$ satisfiable instance of Unique Games on a $d$-regular graph $G$ with alphabet size $k$. Let $M$ be the adjacency matrix of the label-extended graph of $G$.
Let $W$ be the space spanned by eigenvectors of $M$ that have eigenvalue greater than $(1-\gamma)d$, for $\gamma > 8\epsilon$. There is an algorithm that runs in time $2^{O(\frac{\gamma}{\epsilon}\text{dim}(W))}+\text{poly}(n\cdot k)$ and finds an assignment that satisfies at
least $(1-O(\frac{\epsilon}{\gamma-8\epsilon}+\epsilon))$ fraction of the constraints.
\end{thm}

We also show a similar theorem for the special case where all the constraints are of the form $\Gamma$-Max-Lin and the graph satisfies some special properties. That theorem is then used to derive the above-mentioned algorithm for expanders.

 Our key contribution is a technique of effectively using the full spectrum and eigenvectors of the relevant graph rather than just the first few eigenvalues. Interestingly, prior to this work, researchers had often attempted to take advantage of the full spectrum of a graph in the design of approximation algorithms, but no significant progress was made. We believe that our techniques are of independent interest and could contribute to developing better algorithms for a number of other problems. In a recent breakthrough paper, Arora, Barak and Steurer \cite{ABS} used our main technique as a crucial building block in their sub-exponential time algorithm for arbitrary instances of Unique Games. Consequently, we believe that our results
give evidence that spectral techniques might be a more powerful tool than SDPs for attacking and, potentially even disproving, the Unique Games Conjecture.

\subsection{Comparison to Other Work on Unique Games}

Several polynomial time approximation algorithms using linear and semidefinite programming have been developed for Unique Games on arbitrary graphs (see \cite{KhotUCSP}, \cite{Trevisan}, \cite{GT}, \cite{CharikarMM}, \cite{ChlamtacMM}). These algorithms
start with an instance where the value of the SDP or LP relaxation is
$1 - \epsilon$ and round it to a solution with value $\nu$. Here, value of the game
refers to the maximum fraction of satisfiable constraints. For $\nu > \delta$, this would give
an algorithm to distinguish between the two cases. However, most of these algorithms give good
approximations only when $\epsilon$ is very small ($\epsilon = O(1/\log n)$ or $\epsilon = O(1/\log k)$) and  their approximation guarantee depends on the alphabet size $k$
\footnote{It might be good to think of $k$ as $O(\log n)$ since this is the range of interest for most reductions.}.
For constant $\epsilon$, only the algorithm in \cite{CharikarMM} gives interesting parameters with
$\nu \approx k^{-\epsilon/(2-\epsilon)}$. We refer the reader to \cite{CharikarMM} for a comparison
of parameters of various algorithms.

 For some special cases of graphs, it has been shown that there are efficient algorithms that solve Unique Games. One such example
is when the constraint graph is a \textit{spectral expander}. In that case, the authors in \cite{AKKTSV} (and later \cite{MM} with improved parameters) showed that
one can find a highly
satisfying assignment (with, say, $\nu \geq 90\%$) in polynomial time. Notably, the approximation guarantee of their algorithm depends on the
expansion parameters of the graph rather than the alphabet size. As another example,
the authors in \cite{AIMS} presented efficient algorithms for constraint graphs with large {\em local expansion}.

Compared to the previous papers, our algorithm uses purely spectral methods as opposed to linear or semidefinite programming and runs in time that depends on spectral properties of the constraint graph or the label-extended graph. The approximation guarantee does not depend on the size of the graph or the size of the alphabet and thus our algorithm always finds a good assignment when one exists.
We note that, in an unpublished manuscript, Kolla and Tulsiani \cite{KT} developed a {\em spectral} algorithm that runs in polynomial time and finds highly satisfying
assignments on expander constraint graphs.

\subsection{Overview of Our Techniques: Recovering Solutions by Spectral Methods}

Our basic approach for the Unique Games algorithm is exhaustive search in a subspace spanned by several eigenvectors of a graph associated with the Unique Games instance. We identify a ``good'' subspace which contains a vector that
 ``encodes'' a highly satisfying assignment and then exhaustively search for this vector in (a discretization of) the subspace as we explain below.

Let $\mathcal{U}$ be our Unique Games instance which is $(1-\epsilon)$ satisfiable. One can think of this instance as ``coming from'' a completely satisfiable instance $\mathcal{\tilde{U}}$ as follows: an adversary on input $\mathcal{\tilde{U}}$,
 picks an $\epsilon$ fraction of edges and changes the constraints on those edges so that the resulting instance $\mathcal{U}$ becomes $(1-\epsilon)$ satisfiable.
 Let $M$ be (the adjacency matrix of) the graph corresponding to $\mathcal{U}$ and $\widetilde{M}$ (the adjacency matrix of) the graph corresponding to $\mathcal{\tilde{U}}$. Assume, for simplicity, that these graphs are $d$-regular. We first observe that a characteristic vector $y_\alpha$ of a perfectly satisfying assignment $\alpha$ for $\mathcal{\tilde{U}}$ is an eigenvector of $\widetilde{M}$ with eigenvalue $d$.
 
 We next consider a space $W$ which is the span of the eigenvectors of $M$ with
eigenvalue very close to $d$ and show that every characteristic vector $y_\alpha$ is close to some vector in $W$ (in $\ell_2$  norm).

 Our algorithm simply looks at a set
$\mathcal{S} \subseteq W$ of appropriately many candidate vectors and ``reads-off"
an assignment.

{\large $\texttt{Recover-Solution}_\mathcal{S}(\mathcal{U})$}
\begin{itemize}
\item For each $x \in \mathcal{S}$, construct an assignment $L_x$ by assigning
to each vertex $u$, the index corresponding to the largest entry in the
block $(x_{u1}, \ldots, x_{uk})$ i.e. $L_x(u) = \arg \max_i{x_{ui}}$.
\item Out of all assignments $L_x$ for $x \in \mathcal{S}$, choose the one satisfying
the maximum number of constraints.
\end{itemize}

To choose $\mathcal{S}$,we identify
a set of \textit{nice} vectors $\mathcal{N}\subseteq W$ such that the above algorithm works for any vector $v$
close to some vector in $\mathcal{N}$. These are going to be the vectors that are close to the assignment eigenvectors.
 We then construct a set $\mathcal{S} \subseteq W$ of test vectors
such that at least one vector  $v\in \mathcal{S}$ is close to a vector in $\mathcal{N}$. $\mathcal{S}$ is simply an epsilon-net for $W$. Lastly, we go over all vectors in $\mathcal{S}$ until we find $v$. The running time of the algorithm will depend on the size of $\mathcal{S}$ which, for an appropriately defined epsilon-net, it is exponential in the dimension of $W$.

Certain ``simple'' graphs, including expander graphs and the Khot-Vishnoi instance, have only
a few eigenvalues close to $d$ and thus the dimension of $W$ is small. Consequently, exhaustive enumeration in the subspace spanned by the
corresponding eigenvalues will quickly find a good-enough assignment.

\subsection{Organization of the Paper} The rest of the paper is organized as follows: in section~\ref{sec:prelim} we give some preliminary notation and definitions that will be used in the rest of the paper. Section~\ref{sec:main} contains the proofs of the main theorem above as well as of a theorem for the $\Gamma$-Max-Lin case. More specifically, we prove the main theorem in subsection~\ref{sec:31}. The proof of the theorem for the $\Gamma$-Max-Lin case appears in subsection~\ref{subsec:gamma_max_lin1}. Subsection~\ref{subsec:non_reg} contains the generalization of the main theorem for non-regular graphs. In section~\ref{sec:examples}, we analyze our algorithm when run on the Khot-Vishnoi instance.

\section{Preliminaries}
\label{sec:prelim}
\subsection{Spectra of Graphs}

We remind the reader that for a graph $G$, the adjacency matrix
$A=A_G$ is defined as
\begin{align*}
A_G(u,v) = \left\{
\begin{array}{ll}
1 & \text{if}~(u,v) \in E\\
0 & \text{if}~ (u,v) \notin E\\
\end{array}
\right.
\end{align*}
We will also be dealing with weighted graphs, where edge $(u,v)$ has weight $w_{uv}\geq 0$ and the adjacency matrix is defined as $A_G(u,v)=w_{uv}$.
We will also assume, w.l.o.g. that $w_{uv} \leq 1$ (if not, we just re-scale the weights of the edges by the maximum weight).
If the graph has $n$ vertices, $A_G$ has $n$ real eigenvalues
$\lambda_1 \geq \lambda_2 \geq \cdots \geq \lambda_n$. We can always choose $n$ eigenvectors $\gamma_1,\cdots{},\gamma_n$ such that $\gamma_i$ has eigenvalue $\lambda_i$ which form an orthonormal basis of
$\mathbb{R}^n$. We note that if the graph is $d$-regular (the total weight of edges adjacent to every node is $d$) then the
largest eigenvalue is equal to $d$ and the corresponding unit length eigenvector
is the (normalized) all-one's vector.\\

For a graph $G$, let $D$ be the diagonal matrix with diagonal entry $D(u,u) =\sum_{v:(u,v)\in E}w_{uv}$ equal to the degree of node $u$, namely the sum of the weights of edges adjacent to $u$. The Laplacian of $G$ is defined as follows:
 $$L_G =D-A_G$$

If the graph has $n$ vertices, $L_G$ has $n$ real eigenvalues
$0=\mu_1 \leq \mu_2 \leq \cdots \leq \mu_n$. We can always choose $n$ eigenvectors $\gamma_1,\cdots{},\gamma_n$ such that $\gamma_i$ has eigenvalue $\mu_i$ which form an orthonormal basis of
$\mathbb{R}^n$. We note that 0 is always an eigenvalue with corresponding unit length eigenvector the (normalized) all-one's vector. Moreover, if and only if the graph has $k$ connected components, then $L_G$ has $k$ eigenvalues equal to zero.\\

\subsection{The Label-Extended Graph}
For a given instance of Unique Games on a constraint graph $G=(V,E)$, let $A$ be the adjacency matrix of $G$ and let $M$ denote the $nk \times nk$ matrix such that the $k \times k$
block $M_{uv}$ is equal to $w_{uv}\Pi_{uv}$. We use $\Pi_{uv}$ to denote the $k\times k$ matrix of the permutation $\pi_{uv}$. Namely,

\[ \Pi_{uv}(i,j) = \left\{ \begin{array}{ll}
         1 & \mbox{if $\pi_{uv}(i)=j$}\\
        0 & \mbox{otherwise}\end{array} \right. \]

Note that $\pi_{uv} = (\pi_{vu})^{-1}$ implies $\Pi_{uv} ={\Pi_{vu}}^T$ therefore $M$ is symmetric.
 We can view $M$ as the adjacency matrix of the {\em Label-Extended} graph for that instance of
Unique Games. Let $k$ denote the size of the alphabet. We will denote this instance by $\mathcal{U}= (G,M,k)$.\\

\begin{definition}(Characteristic vector of a labeling)
For any labeling of the vertices $L=\{L_u\}_{u\in V}$ with $L_u \in [k]$ we define the \textit{characteristic vector} of that labeling as the
 $nk$ dimensional vector $y^{L}$ with
 \[y^{L}(u,i) = \left\{ \begin{array}{ll}
1 &\mbox{if $i=L_u$}\\
        0 & \mbox{otherwise}\end{array} \right. \]

        We will often normalize such vectors to make them unit vectors. In this case, we have
        \[\tilde{y}^{L}(u,i) = \left\{ \begin{array}{ll}
\frac{1}{\sqrt{n}} &\mbox{if $i=L_u$}\\
        0 & \mbox{otherwise}\end{array} \right. \]
        With some abuse of notation we will refer to both vectors as characteristic vectors of labelings whenever it is clear from the context.
\end{definition}

The following is an important observation:

\begin{observation}\label{obs:evectors}
Assume that a certain instance $\mathcal{U}$ is completely satisfiable. Namely, assume that there is a labeling $L=\{L_u\}_{u\in V}$ with $L_u \in [k]$ such that, for all $(u,v)\in E$ we have $\pi_{uv}(L_u)=L_v$. Then the characteristic vector $y^{L}$ is a eigenvector of $M$ with eigenvalue $d$, if $G$ is a $d$-regular graph. It is also an eigenvector of the laplacian $L_M$ with eigenvalue $0$.
\end{observation}

\subsection{$\Gamma$-Max-Lin Instances}
Let $\Gamma$ be an abelian group. A $\Gamma$-Max-Lin instance of Unique Games on graph $G(V,E)$ has for each edge $(u,v)\in E$ a constraint of the form $x_u - x_v = c_{uv}$,
where $x_u, x_v$ are variables taking values in $\Gamma$ and $c_{uv} \in \Gamma$. The alphabet $k$ is the size of the group $\Gamma$. We note that if a pair of labels $(L_u,L_v) \in \Gamma \times \Gamma$ satisfies a constraint $x_u - x_v = c_{uv}$ of the above form, then for all $i \in \Gamma$, the pair $(L_u +i,L_v+i) \in \Gamma \times \Gamma$ also satisfies the constraint. This implies that any labeling $\{L_u\}_{u \in V}$ of the vertices of $G$ is shift-invariant. Namely, for all $i\in \Gamma$, the labelings $\{L_u\}_{u \in V}$ and $\{L_u+i\}_{u \in V}$ satisfy the same set of constraints.

\subsection{The Spectrum of Cayley Graphs}
For the purposes of this paper, we will use a generalized definition for Cayley graphs. For background on the standard definition and properties of Cayley graphs see, for example, \cite{Kaski,Spi_Notes}.

\begin{definition}\label{def:Cayley}(Adjacency matrix of a Cayley graph)
Let $\Omega$ be a group with $n$ elements. The vertex set of a Cayley graph $C(\Omega)$ is the group $\Omega$. The adjacency matrix of $C(\Omega)$ is the $n \times n$ matrix $A_{C(\Omega)}$
defined as follows: For $g_1,g_2\in \Omega$, $A_{C(\Omega)}(g_1,g_2) =f(g_1-g_2)$, where ``$-$'' refers to the inverse group operation and $f:\Omega \rightarrow {\mathbb{R}}_{+}$. Note that the above definition allows
$C(\Omega)$ to have weighted edges, with the only constraint being that the weight of an edge between $g_1$ and $g_2$ depends on their group-theoretic difference.
\end{definition}

We note that the standard definition of a Cayley graph, given a set $S$ of generators, is a special case of the above, where $f(g_1-g_2) =1$ if $g_1-g_2 \in S$ and $0$ otherwise.

For what follows, we will only be interested in $\Omega$ being abelian. In the rest of this section, for simplicity, we will identify a vector $f$ in $\mathbb{R}^n$ with the function $f:\Omega \rightarrow \mathbb{R}$ and we will use the two interchangeably.\\

 It is well known, that for the abelian case $A_{C(\Omega)}$ has an eigenbasis that consists of the $n$ group-theoretic {\em characters} $\chi_g, g\in \Omega$. We remind the reader that every function $f: \Omega \rightarrow \mathbb{R}$ can be written as follows:
$f(x) = \sum_{g\in \Omega}\hat{f}(g)\chi_g(x)$ where $\hat{f}(g)$ is the \textit{fourier coefficient} that corresponds to character $\chi_g$.
For more background on characters and fourier analysis we refer the reader to standard algebra textbooks and surveys see, for example \cite{art}, \cite{nati}, \cite{Kaski}.
\\
Since definition~\ref{def:Cayley} is slightly different than the usual, we present a proof of the following lemma for completeness.

\begin{lemma}\label{lem:cayley}
Let $C(\Omega)$ be  Cayley graph of an abelian group $\Omega$, as in definition~\ref{def:Cayley}, and $A_{C(\Omega)}$ be its adjacency matrix.
Then for every $g\in\Omega$, $A_{C(\Omega)}\chi_g=|\Omega|\hat{f}(g)\chi_g$. Namely, every character $\chi_g$ is an eigenvector of $A_{C(\Omega)}$ and $|\Omega|\hat{f}(g)$ is its corresponding eigenvalue. Here, with $\hat{f}$ we denote fourier coefficients.
\end{lemma}
\begin{proof}
It is enough to prove the equality for each entry of $A_{C(\Omega)}\chi_g$, that is for every vertex $x\in \Omega$ we will show
$\{A_{C(\Omega)}\chi_g\}_x = |\Omega|\hat{f}(g)\chi_g(x)$. We calculate:
$$\sum_{y\in \Omega}A_{C(\Omega)}(x,y)\chi_g(y) =\sum_{y\in \Omega}f(x-y)\chi_g(y)=\sum_{z\in \Omega}f(z)\chi_g(x)\chi_g(-z) =|\Omega|\hat{f}(g)\chi_g(x)$$
\end{proof}

\subsection{The Khot-Vishnoi Graph}

In \cite{KV}, the authors considered the following family of graphs:\\

For parameters $n$ and $\epsilon$, let $N=2^n$ and $n=2^k$. Denote by $\mathcal{F}$ the family of all boolean functions on $\{-1,1\}^k$. For
$f,g \in \mathcal{F}$ define the product $fg$ as $(fg)(x) =f(x)g(x)$. Let $H=\{\chi_S | S\subseteq [k]\}$ be the set of characters of the group $\mathbf{F}_2^n$. Consider the equivalence relation $\equiv $ on $\mathcal{F}$ defined as
 $f\equiv g$ iff there is an $S\subseteq [k]$ such that $f=g\chi_S$. This relation partitions $\mathcal{F}$ into equivalence classes $\mathcal{P}_1,\cdots,\mathcal{P}_m$,
 where $m=\frac{N}{n}=\frac{2^n}{n}$. For each equivalence class $\mathcal{P}_i$, we could pick an arbitrary representative $p_i\in \mathcal{P}_i$,
so that $$\mathcal{P}_i=\{p_i\chi_S|S\subseteq [k]\}$$

For $0<\epsilon\leq 1$ let $f\in_\epsilon \mathcal{F}$ denote a random boolean function on $\{-1,1\}^k$
where for every $x\in \{-1,1\}^k$, independently, $f(x)=1$ with probability $1-\epsilon$ and $f(x) =-1$ with probability $\epsilon$.
For the given parameter $\epsilon$ and boolean functions $f,g \in \mathcal{F}$ let $$\mathbf{wt}_\epsilon(\{f,g\}):=\textbf{Pr}_{f' \in_{1/2}\mathcal{F}, \mathbf{\mu} \in_{\epsilon}\mathcal{F}}\big[(\{f=f'\}\wedge \{g=f'\mathbf{\mu}\})\vee(\{g=f'\}\wedge \{f=f'\mathbf{\mu}\})\big]$$\\

The $\KV_{n,\epsilon}$ constraint graph with parameters $n$ and $\epsilon$ can now be defined to have vertex set $V=\{\mathcal{P}_1,\cdots,\mathcal{P}_m\}$. For every $f\in \mathcal{P}_i$ and $g\in \mathcal{P}_j$, there
is an edge between the vertices $\mathcal{P}_i$ and $\mathcal{P}_j$ with weight
$\mathbf{wt}_\epsilon(\{f,g\})$\\

We next describe the set of constraints. The set of labels for the instance will be $2^{[k]}:= \{S:S\subseteq [k]\}$, i.e. the set of labels $[n]$ is identified with the set $2^{[k]}$ (and, thus, $n=2^k$). This identification will be used from now on. Let $f \in \mathcal{P}_i$ with $f = p_i\chi_S$ and $g \in \mathcal{P}_j$ with $g = p_i\chi_T$ for some $S,T \subseteq [k]$. The constraint $\pi_e$ for the edge $e\{f,g\}$ corresponding to the pair of functions $f,g$ can now be defined as $$\pi_e^{\mathcal{P}_i}(T\vartriangle U) = S\vartriangle U, \quad \forall U \subseteq [k]$$
Here, $\vartriangle$ is the symmetric difference operator on sets.\\

 We denote the label-extended graph of the instance described above with $\widetilde{\KV}_{n,\epsilon}$. The graph $\widetilde{\KV}_{n,\epsilon}$ has vertex set $\{-1,1\}^n$. Between two vertices $f,g \in \{-1,1\}^n$ there is an edge with weight $n\cdot \mathbf{wt}_\epsilon(\{f,g\})$.\\

It will be useful, for the purposes of this paper, to consider an equivalent definition of $\KV_{n,\epsilon}$ and $\widetilde{\KV}_{n,\epsilon}$, by translating to the $\{0,1\}$ language.
Let $H_n = \{0,1\}^n$ be the group $\mathbf{F}_2^n$ with addition (modulo 2) as group operation. Let $n=2^k$. Then the set $H$ is the {\em Hadamard} code on $n$ bits. Namely, $$H =\left\{h_y \in \{0,1\}^n, y \in \{0,1\}^k |\text{s.t. } h_y(x) =\langle x,y \rangle, x \in \{0,1\}^k \right\}$$ Note that $|H| =n$ and $H$ is a subgroup of $H_n$.\\
 The graph $\KV_{n,\epsilon}$ is just a Cayley graph of the quotient group of $H_n$ by $H$, namely of the group $Q=H_n\diagup H$. The vertex set $V=\{\mathcal{P}_1,\cdots,\mathcal{P}_m\}$ consists of
the $\frac{N}{n}$ cosets of $H$. For every $f$ in the coset $\mathcal{P}_i$ and for every $g$ in the coset $\mathcal{P}_j$, there
is an edge between the vertices $\mathcal{P}_i$ and $\mathcal{P}_j$ with weight
$\mathbf{wt}_\epsilon(\{f,g\})$. For simplicity of notation, we will identify the cosets with their representatives, i.e. if $x=\mathcal{P}_i$ is a vertex of $\KV_{n,\epsilon}$, then we will use $x$ to refer both to the coset $\mathcal{P}_i$ as well as to its representative $p_i$.\\

 The total weight of all the edges between two cosets $x=\mathcal{P}_i$ and $y=\mathcal{P}_j$ is
$$A(x,y) =\sum_{f \in \mathcal{P}_i, g\in \mathcal{P}_j}\epsilon^{|f+g|}(1-\epsilon)^{n-|f+g|}= \sum_{h_1,h_2 \in H}\epsilon^{|x+h_1+y+h_2|}(1-\epsilon)^{n-|x+h_1+y+h_2|}$$
Note that the graph is $n$-regular, namely the total weight of edges adjacent to any node $x$ is
\begin{eqnarray*}
 &&\sum_{y\in Q} \sum_{h_1,h_2 \in H}\epsilon^{|x+h_1+y+h_2|}(1-\epsilon)^{n-|x+h_1+y+h_2|} =n\sum_{y\in Q}\sum_{h_2 \in H}\epsilon^{|x+y+h_2|}(1-\epsilon)^{n-|x+y+h_2|}\\
&=&n \sum_{z\in H_n}\epsilon^{|z|}(1-\epsilon)^{n-|z|}=n\sum_{i=0}^{n}{n \choose i }\epsilon^i (1-\epsilon)^{n-i} =n(1-\epsilon+\epsilon)^n=n
\end{eqnarray*}

The label-extended graph $\widetilde{\KV}_{n,\epsilon}$ is a Cayley graph of $H_n$, since for every two vertices $u,v \in \{0,1\}^n$ we have that the corresponding entry of the adjacency matrix of $\widetilde{\KV}_{n,\epsilon}$ depends only on the hamming distance of $u$ and $v$, or, in other words, on $u-v=u+v$ (modulo 2). Moreover, $\widetilde{\KV}_{n,\epsilon}$ can be described as an ``$\epsilon$-perturbed'' version of the standard hypercube graph: the weight between two vertices $u,v \in \{0,1\}^n$ is proportional to the probability that if one starts from the $n$-bit string $u$ and flips each one of its bits independently with probability $\epsilon$, then one gets the $n$-bit string $v$.  \\
The following appears in \cite{KV}. We give the informal statement here, for simplicity. We refer the reader to \cite{KV} for the full statement and details.
\begin{claim}(Integrality Gap for Unique Games, Informal Statement) \label{cl:KV}
Let $n$ be an integer and $\epsilon>0$ be a parameter. Then for the instance of Unique Games described above, with label-extended graph $\widetilde{\KV}_{n,\epsilon}$ the following hold:
\begin{itemize}
\item Every labeling $L:V \rightarrow [n]$ satisfies {\em at most} a $\frac{1}{n^{\epsilon}}$ fraction of the total weight of the edges.
\item The standard SDP relaxation for Unique Games has objective value greater than $1-9\epsilon$.
\end{itemize}
\end{claim}

\subsection{Additional Notation}
For simplicity of the presentation, we will assume that $G$ is a regular graph, namely $\sum_v w_{uv} =d$ for all nodes $u$. At the end of the next section, we show that the results can easily be generalized to
non-regular graphs.

For the sake of convenience, in the rest of the paper we will often use the term \textit{eigenspace} to refer to the space spanned by a set of eigenvectors that don't necessarily have the same eigenvalue. The relevant set of eigenvectors that span this space will always be clear from the context. We use $\text{poly}(n)$ to refer to some polynomial function in $n$.

We also use the following notation for the time needed to compute eigenvalues and eigenvectors:
\begin{definition}\label{defn:TW}
Let $\mathcal{U}=(G,M,k)$ be an instance of Unique Games on a graph $G$ on $n$ nodes and with alphabet size $k$. Let $M$ be the adjacency matrix of the label-extended graph of $G$. Let
$W$ be some space spanned by eigenvectors of $M$. We denote by $T_W(M)$ the time needed to compute an eigenbasis and the corresponding eigenvalues of $W$. Note that $T_W(M)$ is polynomial in the dimension of $M$, namely $T_W(M)= \text{poly}(n\cdot k)$. \\
\end{definition}

We also note that in the rest of this paper we use the terms ``assignment'' and ``labeling'' interchangeably.

\section{Recovering Solutions by Spectral Methods}
\label{sec:main}
 In this section, we will show how, given a $(1-\epsilon)$ satisfiable instance of Unique Games $\mathcal{U}= (G,M,k)$ on graph $G$ and with alphabet size $k$, the eigenvectors of the label-extended graph $M$ may
be used
to recover good assignments. Specifically, we show the following:
\begin{theorem}(Main)\label{thm:main}
Let $\mathcal{U}=(G,M,k)$ be a $(1-\epsilon)$ satisfiable instance of Unique Games on a $d$-regular graph $G$ on $n$ nodes with alphabet size $k$. Let $M$ be the adjacency matrix of the label-extended graph of $G$ as above. Let
$W$ be the space spanned by eigenvectors of $M$ that have eigenvalue greater than $(1-\gamma)d$, for $\gamma > 8\epsilon$. There is an algorithm that runs in time $2^{O(\frac{\gamma}{\epsilon}\text{dim}(W))}+\text{poly}(n\cdot k)$ and finds an assignment that satisfies at
least $(1-O(\frac{\epsilon}{\gamma-8\epsilon}+\epsilon))$ fraction of the constraints.
\end{theorem}

In particular, the theorem implies that for small enough $\epsilon$, for every $1-\epsilon$ satisfiable instance of Unique Games that satisfies the assumptions of the theorem,
one can find an assignment that satisfies more than 90 percent of the constraints in time that depends on the spectral profile of $M$. We also remark the following:

\begin{remark}
If the algorithm of theorem~\ref{thm:main} fails to find an assignment that satisfies at least $(1-O(\frac{\epsilon}{\gamma-8\epsilon}+\epsilon))$ fraction of the constraints, then for the input $\mathcal{U}$, every assignment satisfies less than $(1-\epsilon)$ fraction of the constraints. Here, the constant in the $O(\cdot)$ notation is the same as the one guaranteed by the theorem above.
\end{remark}

We also consider the special case where the constraints are arbitrary
$\Gamma$-Max-Lin. For such constraints, we prove theorem~\ref{thm:gamma_max_lin1} below.
\begin{theorem}\label{thm:gamma_max_lin1}
Let $\mathcal{U}=(G,M,k)$ be a $(1-\epsilon)$ satisfiable $\Gamma$-Max-Lin instance of Unique Games on a $d$-regular graph $G$ on $n$ nodes with alphabet size $k$. Let
$S_{(1-\gamma)}$ be the space spanned by eigenvectors of $G$ that have eigenvalue greater than $(1-\gamma)d$. Assume moreover that every unit-length vector $\phi \in S_{(1-\gamma)}$ has $\norm{\phi}_{\infty}\leq \frac{C}{\sqrt{n}}$ for some constant $C$.
Then, for $\gamma=\Omega(\sqrt{\epsilon})$, there is an algorithm that runs in time $2^{O(k\cdot D_S)}+\text{poly}(n\cdot k)$ and finds an assignment that satisfies at
least $(1-\zeta)$ fraction of the constraints for some $\zeta = O(\sqrt{\epsilon})$. Here $D_S$ denotes the dimension of $S_{(1-\gamma)}$ and the constant in the $\Omega(\cdot)$ notation depends on $C$.
\end{theorem}

Assuming that $T_{S_{(1-\gamma)}}(M)<2^{O(k\cdot D_S)}$ the algorithm in theorem~\ref{thm:gamma_max_lin1} has running time that solely depends on the spectral profile of $G$, since roughly, the space of eigenvectors of $M$ with large eigenvalue has dimension equal to $k$ times the dimension of the corresponding eigenspace of $G$.\\

The result for expander graphs as it appears in \cite{KT} can be derived as a corollary of theorem~\ref{thm:gamma_max_lin1}.

\begin{corollary}\label{cor:expander}(Unique Games are Easy on Expanders)
Let $\mathcal{U}=(G,M,k)$ be a $(1-\epsilon)$ satisfiable $\Gamma$-Max-Lin instance of Unique Games. Assume moreover that $G$ is a $d$-regular spectral
expander on $n$ nodes. Namely, the second eigenvalue of the adjacency matrix of $G$ is $\lambda \leq (1-\gamma)d$, for $\gamma=\Omega(\sqrt{\epsilon})$.
Then, there is a polynomial time algorithm that finds an assignment that satisfies at least $(1-\zeta)$ fraction of the constraints for some $\zeta =O(\sqrt\epsilon)$.
\end{corollary}
\begin{proof}
The eigenspace $S_{(1-\gamma)}$ of theorem~\ref{thm:gamma_max_lin1} consists solely of the all 1's vector. The $\ell_{\infty}$ norm assumption is trivially satisfied with $C=1$. Then, the conclusion of the theorem implies that for $\gamma=\Omega(\sqrt{\epsilon})$,
there is an algorithm that runs in time $2^{O(k\cdot 1)}+T_{S_{(1-\gamma)}}(M)=2^{O(k)}+\text{poly}(n\cdot k)$ and finds an assignment that satisfies at least an $1-O(\sqrt{\epsilon})$ fraction of the constraints. By the statement of the Unique Games Conjecture, it is enough to consider $k$ to be at most logarithmic in $n$. Therefore the algorithm runs in $2^{O(\log n)}+\text{poly}(n\cdot k)=\text{poly}(n)$ time.
\end{proof}

We remark that the $\Gamma$-Max-Lin requirement is necessary for corollary~\ref{cor:expander}. The proof fails to produce similar guarantee for the general case. This is due to the fact that the spectral properties of label-extended graphs that correspond to arbitrary Unique Games are poorly understood.\\

\subsection{Proof of Theorem~\ref{thm:main}}\label{sec:31}

We prove theorem~\ref{thm:main} by combining two facts. First we show that characteristic vectors of good labelings have large projection onto the eigenspace of $M$ spanned by eigenvectors with large eigenvalues. Then we show that if $M$ has at most $D$ large eigenvalues then we can find a vector close (in $\ell_2$-norm) to such a characteristic vector of a labeling in time exponential in $D$, by looking at an appropriate epsilon-net for this eigenspace. Our proof is inspired by the
approach that the authors in the unpublished manuscript \cite{KT} used in order to recover satisfying assignments on expanders.

\textbf{Proof Overview.}
Assume we are given a $(1-\epsilon)$ satisfiable instance $\mathcal{U} =(G,M,k)$ that satisfies the assumptions of theorem~\ref{thm:main}.
We define a completely satisfiable game $\mathcal{\widetilde{U}}=(G,\widetilde{M},k)$ (with value 1) that ``corresponds'' to $\mathcal{U}$ as follows:
\begin{definition}(Completion of a Game)
Let $\mathcal{U} =(G,M,k)$ be as above and let $\mathcal{L}=\{L(u)\}$ be an assignment that satisfies $(1-\epsilon)$ fraction of the constraints.
 Let $e=(u,v)$ be an edge with constraint $x_v=\pi_{uv}(x_u)$ that is not satisfied by this assignment, with labels $L(u)$, $L(v)$  on its endpoints.
Construct a new game by changing the constraint on every such edge $e=(u,v)$ by replacing $\pi_{uv}$ with a permutation $\tilde{\pi}_{uv}$ for which $\tilde{\pi}_{uv}(L(u))=L(v)$.
Let $\mathcal{\widetilde{U}}=(G,\widetilde{M},k)$ be this new completely satisfiable game. We say that $\mathcal{\widetilde{U}}$ is a
``completion'' of $\mathcal{U}$.
\end{definition}
For the sake of the proof, we will assume that an \textit{almost} satisfiable instance $\mathcal{U}=(G,M,k)$ is constructed as follows:
\begin{itemize}
\item Let $\mathcal{\widetilde{U}}=(G,\widetilde{M},k)$ be a completion of $\mathcal{U}$.
\item Let an adversary pick the $\epsilon$ fraction of edges that were unsatisfied in $\mathcal{U}$ and change their constraints back to the original ones. We can now think of $M$ as a ``perturbation'' of $\widetilde{M}$ and $\mathcal{U}$ as a ``perturbed'' game of $\mathcal{\widetilde{U}}$.
\end{itemize}

Let $W$ be the span of the eigenvectors of $M$ with
eigenvalue at least $(1-\gamma)d$, for some $\gamma > 8\epsilon$. The algorithm simply looks at a set
$\mathcal{S} \subseteq W$ of appropriately many candidate vectors and ``reads-off"
an assignment. We describe later on how to choose the set $\mathcal{S}$.\\

{\large $\texttt{Recover-Solution}_\mathcal{S}(\mathcal{U})$}
\begin{itemize}
\item For each $x \in \mathcal{S}$, construct an assignment $L_x$ by assigning
to each vertex $u$, the index corresponding to the largest entry in the
block $(x_{u1}, \ldots, x_{uk})$ i.e. $L_x(u) = \arg \max_i{x_{ui}}$.
\item Out of all assignments $L_x$ for $x \in \mathcal{S}$, choose the one satisfying
the maximum number of constraints.
\end{itemize}
To choose $\mathcal{S}$, we look at the highest eigenvectors for the matrix $\widetilde{M}$. Those are the \textit{assignment} eigenvectors, namely the characteristic vectors of the (perfectly) satisfying assignments of $\mathcal{\widetilde{U}}$
(since they all have eigenvalue equal to $d$). We will first observe that every such eigenvector $y$ is close to some vector in $W$, and the length of the projection of $y$ onto $W$ depends on $\epsilon, \gamma$. We then identify
a set of \textit{nice} vectors $\mathcal{N}\subseteq W$ such that the above algorithm works for any vector $v$
close to some vector in $\mathcal{N}$. These are going to be the vectors that are close to the assignment eigenvectors.
 We then construct a set $\mathcal{S} \subseteq W$ of test vectors
such that at least one vector  $v\in \mathcal{S}$ is close to a vector in $\mathcal{N}$. $\mathcal{S}$ is simply an epsilon-net for $W$. Lastly, we go over all vectors in $\mathcal{S}$ until we find $v$. The running time of the algorithm will depend on the size of $\mathcal{S}$ which, for an appropriately defined epsilon-net, it is exponential in the dimension of $W$.

\subsubsection{Eigenspaces and Labelings}

Our first step is to identify an appropriate eigenspace $W$ of $M$ such that characteristic vectors of good labelings have large projection onto $W$.

For the matrix $\widetilde{M}$, let $\mathcal{L}=\{L(u)\}_{u\in V(G)}$ be a completely satisfying assignment. Define vector $y^{(\mathcal{L})}$ as
\begin{equation*}
{y_{ui}}^{(\mathcal{L})} = \left\{
\begin{array}{ll}
\frac{1}{\sqrt{n}}  & \text{if}~~~ i = L(u) \\
0 & \text{otherwise}
\end{array}\right.
\end{equation*}

It is easy to see that $y^{(\mathcal{L})}$ is an eigenvector of $\widetilde{M}$ with eigenvalue $d$. Since $y^{(\mathcal{L})}$ corresponds to the satisfying assignment $\mathcal{L}$, it can be seen as the characteristic vector of the assignment. We refer to such vectors as the ``assignment'' eigenvectors.
 Our next goal is to show that, for some appropriate choice of $\gamma$, the eigenspace $W$ contains vectors close to an assignment vector.
\begin{claim}(Closeness)\label{cl:closeness}
For every completely satisfying assignment $\mathcal{L}$ there is a unit vector $v_{\mathcal{L}} \in W$ such that $v_{\mathcal{L}}=\alpha y^{(\mathcal{L})}+\beta {y^{(\mathcal{L})}}_\perp$, with $\alpha >0$ and $|\beta|\leq \sqrt{\frac{2\epsilon}{\gamma}}$. Here both ${y^{(\mathcal{L})}}$ and ${y^{(\mathcal{L})}}_\perp$ are unit vectors and ${y^{(\mathcal{L})}}_\perp \perp y^{(\mathcal{L})}$. By taking, for example, $\gamma\geq 200 \epsilon$, we have $|\beta|\leq \frac{1}{10}$.
\end{claim}
\begin{proof}
We can easily see that $(y^{(\mathcal{L})})^TMy^{(\mathcal{L})} \geq d(1-2\epsilon)$.
We can now write $y^{(\mathcal{L})} =a v_{\mathcal{L}}+b (v_{\mathcal{L}})_\perp$, with $a>0$, ${v_{\mathcal{L}}}\in W$ and $({v_{\mathcal{L}}})_\perp \in W^\perp$. We calculate:
$$(1-2\epsilon)d \leq (y^{(\mathcal{L})})^TMy^{(\mathcal{L})} = a^2({v_{\mathcal{L}}})^T M{v_{\mathcal{L}}} + b^2((v_{\mathcal{L}})_\perp)^TM(v_{\mathcal{L}})_\perp \leq a^2 d +b^2(1-\gamma)d$$ from which we get that
$|b| \leq \sqrt{\frac{2\epsilon}{\gamma}}$.

Now, we can in turn express $v_{\mathcal{L}} = \alpha y^{(\mathcal{L})} +\beta {y^{(\mathcal{L})}}_\perp$, where $\alpha = \langle v_{\mathcal{L}}, y^{(\mathcal{L})}\rangle = a = \sqrt{1-b^2} \geq \sqrt{1-\frac{2\epsilon}{\gamma}}$ or, equivalently, $|\beta| =\sqrt{1-\alpha^2} \leq \sqrt{\frac{2\epsilon}{\gamma}}$.
\end{proof}

We have now managed to show that there exists a set of ``nice'' vectors $\mathcal{N} =\{v_{\mathcal{L}}\}\subseteq W$. The next claim shows that, if we knew the $v_{\mathcal{L}}$'s then we could set $\mathcal{S}= \mathcal{N}$ and the algorithm $\texttt{Recover-Solution}_{\mathcal{N}}(\mathcal{U})$ would return a highly satisfying assignment for $\mathcal{U}$.

\begin{claim}\label{cl:uniquemax}
If $x$ is a vector such that $x = \alpha y^{(\mathcal{L})} + \beta {y^{(\mathcal{L})}}_{\perp}$ for
some $y^{(\mathcal{L})}$ with $\alpha > 0$ and ${y^{(\mathcal{L})}}_{\perp} \perp y^{(\mathcal{L})}$, then the coordinate ${x_{u L(u)}}$ is maximum in absolute value
in at least  $(1-\frac{2\beta^2}{\alpha^2})n$ blocks.
\end{claim}

\begin{proof}
Within each block $u$, in
order for coordinate $L(u)$ to be no longer the maximum one, it must
happen that for some $j$
\begin{equation*}
\alpha \frac{1}{\sqrt{n}}+\beta ({y^{(\mathcal{L})}}_{\perp})_{uL(u)}\leq  \beta ({y^{(\mathcal{L})}}_{\perp})_{uj}
\end{equation*}
However, this gives
\begin{eqnarray*}
[({y^{(\mathcal{L})}}_{\perp})_{uj}-({y^{(\mathcal{L})}}_{\perp})_{uL(u)}]^2 &\geq& \frac{\alpha^2}{\beta^2n}\Rightarrow \\
2({y^{(\mathcal{L})}}_{\perp})_{uj}^2+2({y^{(\mathcal{L})}}_{\perp})_{uL(u)}^2 &\geq&\frac{\alpha^2}{\beta^2n}\Rightarrow\\
\norm{({y^{(\mathcal{L})}}_{\perp})_u}^2 \geq ({y^{(\mathcal{L})}}_{\perp})_{uj}^2 +({y^{(\mathcal{L})}}_{\perp})_{uL(u)}^2 &\geq& \frac{\alpha^2}{2\beta^2 n}
\end{eqnarray*}
Since $\norm{{y^{(\mathcal{L})}}_{\perp}} = 1$, this can only happen for at most $\frac{2\beta^2}{\alpha^2}n$
blocks.
\end{proof}

\subsubsection{Eigenspace Enumeration: Finding the Set $\mathcal{S}$ of Test Vectors}\label{subsec:setS}

 Our next step is to show that the search for a good assignment can be done in time exponential to the number of large eigenvalues of $M$, namely the dimension of $W$.

 Since we don't know the $y^{(\mathcal{L})}$ (if we did, we would be done), and the space $W$ contains infinitely many unit vectors, we cannot identify
$\mathcal{N}$. To get around this, we {\em discretize} $W$ with an appropriate epsilon-net.
 If we let $w^{(0)}, \ldots, w^{(\text{dim}(W)-1)}$ be an eigenbasis for $W$.
We define the set $\mathcal{S}$ as
\begin{equation*}
\mathcal{S} ~~=~~ \left\{ v = \sum_{s=0}^{(\text{dim}(W)-1)} \alpha_s w^{(s)} ~~\vert~~
\alpha_s \in \sqrt{\frac{2\epsilon}{\gamma\text{dim}(W)}}\bz, ~ \norm{v} \leq 1\right\}
\end{equation*}
It can be calculated (see, for instance \cite{FO05}) that the number of points in the set $\mathcal{S}$ is at most $2^{O(\frac{\gamma}{\epsilon}\text{dim}(W))}$.

To conclude the proof of theorem~\ref{thm:main}, it remains to show that $\mathcal{S}$ has a vector close to a nice vector $v_{\mathcal{L}} \in \mathcal{N}$. By construction, $\mathcal{S}$ contains at least one vector close to every vector in $\mathcal{N}$ and thus it also contains at least one vector $v$ such that $v = \alpha v_{\mathcal{L}} + \beta ({v_{\mathcal{L}}})_\perp$
for some $\mathcal{L}$ and $|\beta| \leq \sqrt{\frac{2\epsilon}{\gamma}}$. Together with claim~\ref{cl:closeness}, this implies that we can also write  $v = a y^{(\mathcal{L})} +by^{(\mathcal{L})}_\perp$, with $|b|\leq \sqrt{\frac{2\epsilon}{\gamma}}+\sqrt{\frac{2\epsilon}{\gamma}}=2\sqrt{\frac{2\epsilon}{\gamma}}$. Thus, by claim~\ref{cl:uniquemax}, for this vector $v$, $\texttt{Recover-Solution}_{\mathcal{S}}(\mathcal{U})$
 recovers an assignment which agrees with $y^{(\mathcal{L})}$ in $(1-O(\frac{\epsilon}{\gamma-8\epsilon}))$
fraction of the blocks.\\ We will consider that a constrain is violated if one of the two following events happen: either the assignment we recovered does not agree with the initial perfectly satisfying assignment or the constraint is one of the (at most) $\epsilon$ fraction of constraints that were changed. It follows that the assignment which we recovered, violates constraints on edges that have total weight at most $O(\frac{\epsilon}{\gamma-8\epsilon})nd + \epsilon nd$.
Since the total weight of constraints is $nd/2$, theorem~\ref{thm:main} follows. The algorithm runs in $2^{O(\frac{\gamma}{\epsilon}\text{dim}(W))}+T_W(M)=2^{O(\frac{\gamma}{\epsilon}\text{dim}(W))}+\text{poly}(n\cdot k)$ time.

\subsection{Proof of Theorem~\ref{thm:gamma_max_lin1}}
\label{subsec:gamma_max_lin1}

 Let $\theta>0$ such that $\gamma \geq \theta\geq \Omega(\epsilon \gamma)$. Let $W$ be the span of eigenvectors of $M$ with eigenvalue at least $(1-\theta)d$. Let
$Y$ be the span of eigenvectors of $\widetilde{M}$ with eigenvalue at least
$(1-\gamma)d$. In particular, $Y$ contains all the \textit{assignment} eigenvectors, namely the characteristic vectors of the (perfectly) satisfying assignments of $\mathcal{\widetilde{U}}$
(since they all have eigenvalue equal to $d$). Let $Y_\perp$ be the orthogonal complement of $Y$.\\
 For the proof, we will use theorem~\ref{thm:main} together with a bound on the dimension of W. Roughly we show that the dimension of $W$ is at most as large as the dimension of $Y$ which, in turn, is at most $k$ times the dimension of $S_{(1-\gamma)}$. To conclude, we apply theorem~\ref{thm:main} for the eigenspace $W$.\\

 We will assume w.l.o.g that the graphs we are dealing with are connected. Otherwise, we apply the results to each connected component individually.

\begin{definition}\label{def:eigenbasislift}
Let $\Phi$  be an eigenbasis for $G$. For every eigenvector $\phi= (\phi_u)_{u\in V} \in \Phi$ let $\widetilde{\phi}$ to be the $k n$-dimensional
vector $\widetilde{\phi}_{ui}= \phi_u$, for $i=1,\cdots, k$. Due to shift invariance, there exist $k$ satisfying labelings of the unique games instance which we denote by $L_i, i=1,\cdots,k$. Let $y^{(L_i)}$ be the characteristic vector of labeling $L_i$ and let $\mathcal{E}(\phi) =\{\widetilde{\phi}\cdot y^{(L_i)}|\quad \text{for} \quad i=0,\cdots,k-1\}$. Here $\cdot $ denotes entry-wise vector product. We define the following set of $k\times n$ vectors: $\tilde{\Phi}=\bigcup_{\phi\in \Phi}\mathcal{E}(\phi)$.
\end{definition}

We note that the set $\tilde{\Phi}$ consists of orthogonal unit vectors, since every two vectors in $\tilde{\Phi}$ have disjoint support. The next observation immediately follows from the fact that, due to the shift-invariance property of $\Gamma$-Max-Lin instances, $\widetilde{M}$ is the adjacency matrix of $k$ disconnected copies of $G$.

\begin{observation}\label{cl:dimensionD_S}
 $\tilde{\Phi}$ is an eigenbasis of $\widetilde{M}$. Consequently, the dimension of $Y$ is at most $k\cdot D_S$, where $D_S$ is the dimension of $S_{(1-\gamma)}$, as per the assumption of the theorem.
\end{observation}

\begin{lemma}\label{cl:betabound}
Assume $\gamma\geq \theta\geq \Omega(\epsilon \gamma)$ for some appropriately large constant. Then any unit-length vector $w\in W$ can be expressed as $\alpha y + \beta y_{\perp}$ where $y \in Y$ and $y_\perp \in Y_\perp$ are both unit length vectors, and
$|\beta| \leq O(\sqrt{\frac{\theta}{\gamma^3}})$. By taking, for example, $\gamma \geq \Omega(\sqrt[3]{\theta})$ and, consequently, $\gamma \geq \Omega(\sqrt{\epsilon})$, we get $|\beta| \leq \frac{1}{8}$.
\end{lemma}

 Combining claims~\ref{cl:dimensionD_S} and~\ref{cl:betabound}, we can proceed as follows:  From claim~\ref{cl:betabound} we obtain that $W$ has dimension $\text{dim}(W) \leq \text{dim}(Y)$. Otherwise, we would find a
vector orthogonal to all the vectors in $Y$ which cannot be close to their span. From claim~\ref{cl:dimensionD_S} we obtain $\text{dim}(Y)=k\cdot D_S$.\\
To conclude the proof of theorem~\ref{thm:gamma_max_lin1}, we apply theorem~\ref{thm:main} with $W$ being the (at most) $k\cdot D_S$ dimensional eigenspace of $M$ with eigenvalues $\geq (1-\theta)d$.

Finally, it remains to argue about the running time of the algorithm. Since $\text{dim}(W)\leq k \cdot D_S$, the
number of points is exponential in $kD_S$. Hence, the algorithm runs in $2^{O(k\cdot D_S)}+T_W(M)=2^{O(k\cdot D_S)}+\text{poly}(n\cdot k)$ time.

 \subsubsection{Perturbation of Eigenspaces: Proof of Claim~\ref{cl:betabound}}
We next prove claim~\ref{cl:betabound}. Towards this end, we will apply some results from matrix perturbation theory.
To find appropriate $\gamma$ and $\theta$ such that the eigenspaces $W$ and $Y$ are close, we use the following
claim which essentially appears in \cite{DK70} as the sin $\theta$ theorem and was used in \cite{KT}.
\begin{lemma} \label{cl:sintheta}
Let $w$ be a unit length eigenvector of $M$ with eigenvalue
$\lambda \geq (1-\theta)d$, and let $\lambda_s$  denote the largest
eigenvalue of $\widetilde{M}$ which is smaller than $(1-\gamma)d$. Then, $w$ can be
written as $\alpha y + \beta y_{\perp}$ with
$|\beta| \leq \frac{\norm{(\widetilde{M}-M)w}}{(\lambda - \lambda_s)}$. Here $y\in Y$ and $y_{\perp} \in Y^{\perp}$ are unit length vectors.
\end{lemma}

As seen by the previous lemma, in order to prove that the space $Y$ does not change by much due to the perturbation,  we simply
need to bound $\norm{(\widetilde{M}-M)w}$. We will need the fact that $w$ is somewhat
``uniform" across blocks. To formalize this, let $\bar{w}$ be the $n$-dimensional vector such
that $\bar{w}_u = \norm{w_u}$ where $w_u$ is the $k$-dimensional vector
$(w_{u1}, \ldots, w_{uk})^T$. We then show that $\bar{w}$
is very close to a vector in $S_{(1-\gamma)}$.
\begin{claim}\label{cl:bound2}
If $w$ is a unit-length eigenvector of $M$ with eigenvalue more than $(1-\theta)d$ and $\bar{w}$ as above, then we can write $\bar{w}$
as $\bar{w}=a\phi + b \phi_{\perp}$ with $|b| \leq \sqrt{\frac{\theta}{\gamma}}$ and $\phi \in S_{(1-\gamma)}$, $\phi_{\perp} \in S_{(1-\gamma)}^{\perp}$, both having unit length.
\end{claim}
\begin{proof}
Since, $w$ corresponds to a large eigenvalue,
we have that
\begin{equation*}
(1-\theta)d ~\leq~ (w)^TMw
~=~ \sum_{u,v} {w_u}^T A_{uv}\cdot\Pi_{uv} w_v
~\leq~ \sum_{u,v} \norm{w_u} A_{uv} \norm{w_v}
~=~ (\bar{w})^T A \bar{w}
\end{equation*}

Here the second inequality follows from the fact that the operator norm of $\Pi_{uv}$ is at most 1.
Writing $\bar{w}$ as $a{\phi} + b{\phi}_{\perp}$, we get
\begin{align*}
& ~(\bar{w})^T A \bar{w} ~\leq~   a^2 d + b^2 (1-\gamma)d \\
\implies & ~(1-\theta)d ~\leq a^2 d + b^2 (1-\gamma)d
~\implies~ |b| \leq \sqrt{\frac{\theta}{\gamma}}
\end{align*}
\end{proof}

Using the above, and the fact that the matrix $\widetilde{M}$ is only perturbed in $\epsilon$
fraction of the edges, we can now bound $\norm{(\widetilde{M}-M)w}$ as follows.
\begin{claim}\label{cl:numerator-bound}
$\norm{(\widetilde{M}-M)w} \leq O\left(\sqrt{\frac{\theta}{\gamma}}\right)d$
\end{claim}
\begin{proof}
Define the $n \times n$ matrix $R$ as $R_{uv} = w_{uv}\leq 1$ when the block $(\widetilde{M}-M)_{uv}$ has any
non-zero entry, and $R_{uv} = 0$ otherwise. Note that if $(\widetilde{M}-M)_{uv}$ is non-zero, then it
must be the (scaled) difference of two permutation matrices. Thus, for all $v$ we have
$\norm{(\widetilde{M}-M)_{uv}w_v} \leq 2R_{uv}\norm{w_v}$. We calculate:
\begin{align*}
\norm{(\widetilde{M}-M)w} ~=~ \sqrt{\sum_u\norm{\sum_v(\widetilde{M}-M)_{uv}w_v}^2}
  &~\leq~ \sqrt{\sum_u\left(\sum_v\norm{(\widetilde{M}-M)_{uv}w_v}\right)^2} \\
  &~\leq~ \sqrt{\sum_u\left(\sum_v 2R_{uv}\norm{w_v}\right)^2} \\
  &~=~ 2\norm{R \bar{w}}
\end{align*}
To estimate $\norm{R \bar{w}}$, we break it up as
\begin{equation*}
\norm{R \bar{w}} \leq |a|\norm{R \phi} +
                      |b|\norm{R \phi_{\perp}}
\end{equation*}
Since each row of $R$ has total sum of entries at most $d$,
$|b|\norm{R \phi_{\perp}} \leq \sqrt{\frac{\theta}{\gamma}}d$.
Also,
$$\norm{R \phi} = \sqrt{\sum_u|\sum_v R_{uv}\phi_v|^2}\leq \sqrt{\sum_u\left(\sum_v R_{uv}|\phi_v|\right)^2} \leq \frac{C}{\sqrt{n}}\sqrt{\sum_u\left(\sum_v R_{uv}\right)^2}$$
Since $R$ has a total sum of entries at most $\epsilon nd$, this expression is maximized when it
has $d$ 1s in $\epsilon n$ rows. This gives
$\norm{R \phi} \leq C\sqrt{\epsilon}d$.
And, putting everything together we obtain
\begin{equation*}
\norm{(\widetilde{M}-M)w} ~~\leq~~ 2C\sqrt{\epsilon}d + 2\sqrt{\frac{\theta}{\gamma}}d
~~\leq~~ O\left(\sqrt{\frac{\theta}{\gamma}}\right)d
\end{equation*}

We note that the $O(\cdot)$ in the above expression depends on $C$.
We obtained the last inequality by using the assumption of the lemma $\theta\geq \Omega(\epsilon \gamma)$.
\end{proof}

Combining the above bound with claim~\ref{cl:sintheta}, we get that any unit-length
vector $w \in W$ can be  expressed as $\alpha y + \beta y_{\perp}$ where $y \in Y$ and
$|\beta| \leq O(\sqrt{\frac{\theta}{\gamma}}d \cdot \frac{1}{(1-\theta)d - \lambda_s})$.
Recall that $\lambda_s$ was smaller than $(1-\gamma)d$, which implies
\begin{equation}\label{eq:betabound}
|\beta| \leq O\left(\sqrt{\frac{\theta}{\gamma}}\cdot \frac{1}{\gamma-\theta}\right)
\end{equation}

To conclude the proof of claim~\ref{cl:betabound} we take $\theta<< \gamma$.
We calculate:

\begin{equation}\label{eq:betabound2}
|\beta| \leq O\left(\sqrt{\frac{\theta}{\gamma}}\cdot \frac{1}{\gamma-\theta}\right)\leq O\left(\sqrt{\frac{\theta}{\gamma^3}}\right)
\end{equation}

It is sufficient to consider $\gamma \geq \Omega(\sqrt[3]{\theta})$ and observe that this also implies, by the assumption $\theta\geq \Omega(\epsilon \gamma)$ that $\gamma \geq \Omega(\sqrt{\epsilon})$.

\subsection{Generalizing to Non-Regular Graphs}\label{subsec:non_reg}
It remains to show that the above results hold for non-regular graphs. We will consider eigenvectors and eigenvalues of the Laplacian matrix of the label-extended graph. Let $G$ and $M$ as before. Let $D$ be the $nk \times nk$ diagonal matrix with block $D_{uu} = deg(u) \cdot I_k$, where $I_k$ is the $k\times k$ identity matrix.
The Laplacian of the label-extended graph is the matrix $L_M=D-M$. Let $d$ be the average degree of $G$, and therefore of $M$, namely $d=\frac{\sum_{u\in V}d_u}{n}$.\\
We re-state the main theorem~\ref{thm:main} to capture the non-regular case.

 \begin{theorem}
 Let $\mathcal{U}=(G,M,k)$ be a $(1-\epsilon)$ satisfiable instance of Unique Games and
$W$ the eigenspace of $L_M$ with eigenvalues less than $\gamma d$, for $\gamma \geq 8\epsilon$. There is an algorithm that runs in time $2^{O(\frac{\gamma}{\epsilon}\text{dim}(W))}+\text{poly}(n\cdot k)$ and finds an assignment that satisfies at
least $(1-O(\frac{\epsilon}{\gamma-8\epsilon}+\epsilon))$ fraction of the constraints.
\end{theorem}

\begin{proof}

Let $W$ be the span of the eigenvectors of $L_M$ with
eigenvalue less than $\gamma d$, for some $\gamma \geq 8\epsilon$.
We again recover highly satisfying assignments using our algorithm $\texttt{Recover-Solution}_\mathcal{S}(\mathcal{U})$. Similarly to the previous section, the set $\mathcal{S}\subseteq W$ will be a discretization of our new $W$.

All we need to show is the following analog of claim~\ref{cl:closeness}. We use the same notation as in section~\ref{sec:31}.

\begin{claim}
For every completely satisfying assignment $\mathcal{L}$ there is a unit vector $v_{\mathcal{L}} \in W$ such that $v_{\mathcal{L}}=\alpha y^{(\mathcal{L})}+\beta {y^{(\mathcal{L})}}_\perp$, with $|\beta|\leq \sqrt{\frac{2\epsilon}{\gamma}}$. Here $y^{(\mathcal{L})}$ and ${y^{(\mathcal{L})}}_\perp$ are both unit length vectors and ${y^{(\mathcal{L})}}_\perp \perp y^{(\mathcal{L})}$. By taking, for example, $\gamma\geq 200 \epsilon$, we have $|\beta|\leq \frac{1}{10}$.
\end{claim}

\begin{proof}(Of Claim)

 We first observe that the \textit{assignment} eigenvectors $y^{\mathcal{L}}$ are eigenvectors of $L_{\widetilde{M}}$ with eigenvalue $0$.

We next see that

\begin{eqnarray*}
(y^{(\mathcal{L})})^TL_My^{(\mathcal{L})}&=&(y^{(\mathcal{L})})^T(D-M)y^{(\mathcal{L})}=
(y^{(\mathcal{L})})^T(D-\widetilde{M}+\widetilde{M}-M)y^{(\mathcal{L})}\\
&=& (y^{(\mathcal{L})})^T(\widetilde{M}-M)y^{(\mathcal{L})} \leq \frac{2\epsilon 2|E|}{n}=2\epsilon d
\end{eqnarray*}

We can now write $y^{(\mathcal{L})} =a v_{\mathcal{L}}+b (v_{\mathcal{L}})_\perp$, with ${v_{\mathcal{L}}}\in W$ and $({v_{\mathcal{L}}})_\perp \in W^\perp$. We calculate:
$$2\epsilon d \geq (y^{(\mathcal{L})})^TL_My^{(\mathcal{L})} = a^2({v_{\mathcal{L}}})^T L_M{v_{\mathcal{L}}} + b^2((v_{\mathcal{L}})_\perp)^TL_M(v_{\mathcal{L}})_\perp \geq b^2\gamma d$$ from which we get that
$|b| \leq \sqrt{\frac{2\epsilon}{\gamma}}$.

Now, we can in turn express $v_{\mathcal{L}} = \alpha y^{(\mathcal{L})} +\beta {y^{(\mathcal{L})}}_\perp$, where $\alpha = \langle v_{\mathcal{L}}, y^{(\mathcal{L})}\rangle = |a| = \sqrt{1-b^2} \geq \sqrt{1-\frac{2\epsilon}{\gamma}}$ or, equivalently, $|\beta| =\sqrt{1-\alpha^2} \leq \sqrt{\frac{2\epsilon}{\gamma}}$.
\end{proof}

To conclude the proof of the theorem, we observe that the analog of claim~\ref{cl:uniquemax} follows immediately. We define $\mathcal{S}$ as in section~\ref{subsec:setS}, namely
\begin{equation*}
\mathcal{S} ~~=~~ \left\{ v = \sum_{s=0}^{(\text{dim}(W)-1)} \alpha_s w^{(s)} ~~\vert~~
\alpha_s \in \sqrt{\frac{2\epsilon}{\gamma\text{dim}(W)}}\bz, ~ \norm{v} \leq 1\right\}
\end{equation*}

where $w^{(i)}$ are now eigenvectors of our new $W$. The theorem follows.

\end{proof}

\section{Solving Unique Games on the Khot-Vishnoi Instance}\label{sec:KV}
\label{sec:examples}
 In this section, we show that, when run on the Khot-Vishnoi integrality gap instance, our main algorithm runs in quasi-polynomial time and correctly decides the (un)satisfiability of the instance. Namely, we show the following:

\begin{claim}\label{cl:kv}
  Our main algorithm (as in theorem~\ref{thm:main}) when given input the integrality gap instance of Khot and Vishnoi as described in the preliminaries section, with label-extended graph $\widetilde{\KV}_{n,\epsilon}$, runs in time $n^{\text{poly}(\log n)}$ and correctly decides that the instance is highly unsatisfiable. Here $N=2^n$ the number of nodes of $H_n$ and $n$ is the alphabet size.
\end{claim}

The proof goes by showing that the eigenspace $W$ as per theorem~\ref{thm:main} has relatively low dimension. The theorem guarantees that exhaustive search in $W$ would find a good assignment if such assignment existed and consequently, failure to find such an assignment, will correctly classify  the Khot-Vishnoi instance as unsatisfiable.

The rest of this section is devoted to the proof of claim~\ref{cl:kv}.
We use the equivalent definition of $\widetilde{\KV}_{n,\epsilon}$ as a Cayley graph of $H_n$ which is an ``$\epsilon$-perturbed'' version of the hypercube graph. We will need the following claim for the spectrum of $\widetilde{\KV}_{n,\epsilon}$. We denote the adjacency matrix of $\widetilde{\KV}_{n,\epsilon}$ by $M_{\KV}$.

\begin{claim}\label{cl:KVspec}
The following is true for the spectrum of $\widetilde{\KV}_{n,\epsilon}$.
\begin{itemize}
\item The eigenvectors are the characters of the group $H_n$.
\item The eigenvalue that corresponds to $\chi_{\omega}$ is $(1-2\epsilon)^rn$, where $r=|\omega|$ is the hamming  weight of $\omega$, and appears with multiplicity $C_r= {n\choose r}$.
\end{itemize}
\end{claim}

\begin{proof}
The first item above is immediate from the definition of the graph as it appears in the preliminaries section.

For the second item, we just need to calculate the eigenvalues corresponding to $\chi_{\omega}$ or, equivalently, the quadratic form
\begin{eqnarray*}
&&\frac{1}{\sqrt{N}}(\chi_\omega)^TM_{\KV} (\frac{1}{\sqrt{N}}\chi_\omega)=\frac{1}{N}\sum_{x,y \in H_n}\chi_\omega(x)M_{\KV}(x,y)\chi_\omega(y)\\
  &=& \frac{1}{2^n} \sum_{x,y\in H_n}\chi_\omega(x)n \epsilon^{|x+y|}(1-\epsilon)^{n-|x+y|}\chi_\omega(y)\\
  &=& \frac{n}{2^n} \sum_{x,y\in H_n}\chi_\omega(x+y)\epsilon^{|x+y|}(1-\epsilon)^{n-|x+y|}\\
  &=& \frac{n}{2^n}2^n \sum_{z\in H_n}\chi_\omega(z)\epsilon^{|z|}(1-\epsilon)^{n-|z|}\\
&=&n\mathbf{E}_{z \in_\epsilon \mathcal{F}}\big[\chi_\omega(z)\big] = \prod_{i=1}^n\mathbf{E}_{z_i\in \{0,1\}}[(-1)^{\omega_i\cdot z_i}]=n(1-2\epsilon)^{|\omega|}
\end{eqnarray*}
To conclude the proof, we need to argue about the multiplicity of each eigenvalue or, equivalently, the number of characters $\chi_\omega$ of a given hamming weight $r$. It is easily seen that there are $C_r ={n\choose r}$ characters
of $H_n$ that correspond to hamming weight $r$.
\end{proof}

We will directly apply theorem~\ref{thm:main}. For this purpose we need to calculate the dimension
of the eigenspace $W$ of $\widetilde{\KV}_{n,\epsilon}$ for some appropriate $\gamma \geq \Omega(\sqrt{\epsilon})$.

\begin{lemma}\label{lem:dimension}
Let $\gamma = C\sqrt{\epsilon}$ for some sufficiently large constant $C$. The dimension of the eigenspace $W$ of the graph $\widetilde{\KV}_{n,\epsilon}$ with eigenvalues $\geq (1-\gamma)$ is at most
$D_W \leq \frac{1}{\epsilon} n^{O(\frac{1}{\epsilon})}= \text{poly}(n) = \text{poly}(\log N)$.
\end{lemma}
\begin{proof}
 We first need to identify an upperbound on $r$ such that $(1-2\epsilon)^r \leq (1-\gamma)$, for $\gamma = C\sqrt{\epsilon}$.
It is enough to take $r = \frac{\log 1-\gamma}{\log 1-2\epsilon}$ . In order to approximate $r$, we use the Taylor series expansion:
 \begin{equation*}
 \log \frac{1}{1-x} = \sum_{i=1}^\infty \frac{x^i}{i}
 \end{equation*}

And for $\epsilon, \gamma$ small enough we get

$$r = \frac{\log 1-\gamma}{\log 1-2\epsilon}= \frac{\log \frac{1}{1-\gamma}}{\log \frac{1}{1-2\epsilon}} \approx \frac{\gamma}{\epsilon}\leq C\frac{1}{\sqrt{\epsilon}} \leq \tilde{C} \frac{1}{\epsilon}$$ for some appropriately large constant $\tilde{C}$.
The dimension of $W$ can now be bounded, using claim~\ref{cl:KVspec}, as follows: $D_W \leq \sum_{r=0}^{\tilde{C}1/\epsilon}{n \choose r} \leq \frac{1}{\epsilon} n^{O(\frac{1}{\epsilon})}$.
\end{proof}

We conclude by arguing, according to theorem~\ref{thm:main}, that our algorithm will run in time bounded by
$2^{ D_W} = 2^{ \frac{1}{\epsilon}\cdot n^{O(\frac{1}{\epsilon})}}= N^{\frac{1}{\epsilon}\cdot (\log N)^{O(\frac{1}{\epsilon})}}=N^{(\text{poly}(\log N))}$. It will fail to find a highly satisfying assignment and thus decide that the above instance is highly unsatisfiable.
Here we used the fact that for the graph in question, $T_{S_W}(M)\leq 2^{ D_W}$.

\section*{Acknowledgements}

The author would like to thank Madhur Tulsiani for his contribution to the ideas and techniques for developing spectral algorithms for Unique Games. The author would also like to thank the anonymous reviewers for their comments and suggestions.
\bibliographystyle{alpha}
\bibliography{UGbib2}

\end{document}